\documentclass[onecolumn,10pt]{article}
\usepackage{amsfonts}
\usepackage{bbm}
\usepackage{authblk}

\usepackage{geometry}
\usepackage{amsmath}
\usepackage{amssymb}
\usepackage{amsthm}
\usepackage{graphicx}
\usepackage{subfigure}
\usepackage{rotating}
\usepackage{lscape}
\usepackage{paralist}
\usepackage{mathrsfs}
\usepackage{lettrine}
\usepackage[rm,small,compact]{titlesec}
\usepackage{bibentry}
\usepackage[round,authoryear]{natbib}
\usepackage{abstract}
\usepackage{mathptmx}
\usepackage[scaled=0.92]{helvet}
\usepackage{courier}
\usepackage[T1]{fontenc}
\usepackage{enumitem}
\usepackage{titling}
\usepackage{mathtools}
\usepackage{times}
\usepackage{color}
\usepackage{lineno}
\usepackage{longtable}
\usepackage[varg]{txfonts}
\usepackage[titletoc,title]{appendix}
\usepackage{yfonts}
\usepackage{arydshln}
\usepackage{dsfont}

\theoremstyle{plain}
\theoremstyle{definition}
\newtheorem{definition}{\textsc{Definition}}[section]
\newtheorem{lemma}{\textsc{Lemma}}[section]
\newtheorem{theorem}{\textsc{Theorem}}[section]

\newtheorem*{theorem*}{\textsc{Theorem}}

\theoremstyle{definition}

\newtheorem*{assertion*}{\textsc{Assertion}}

\newlist{Axiom}{enumerate}{1}
\setlist[Axiom]{label=Axiom A\arabic*.}

\makeatletter
\renewenvironment{proof}[1][\proofname]{\par
\vspace{-10pt}
    \pushQED{\qed}%
    \normalfont \partopsep=\z@skip \topsep=\z@skip
  \trivlist
  \item[\hskip\labelsep
        \itshape
    #1\@addpunct{.}]\ignorespaces
}{%
    \popQED\endtrivlist\@endpefalse
} \makeatother

 \makeatletter
    \renewcommand{\thefigure}{\ifnum \c@section>\z@ \thesection.\fi \@arabic\c@figure}
    \renewcommand{\thetable}{\ifnum \c@section>\z@ \thesection.\fi \@arabic\c@table}
    \makeatother

\renewcommand{\proofname}{\textsc{Proof}.}

\makeatletter 
\@addtoreset{equation}{section}
\makeatother  

\makeatletter

\newcommand{\Rmnum}[1]{\expandafter\@slowromancap\romannumeral #1@}
\makeatother

\makeatletter
\newcommand*\bigcdot{\mathpalette\bigcdot@{.5}}
\newcommand*\bigcdot@[2]{\mathbin{\vcenter{\hbox{\scalebox{#2}{$\m@th#1\bullet$}}}}}
\makeatother

\DeclareSymbolFont{euex}{U}{euex}{m}{n}
\DeclareMathSymbol{\varint}{\mathop}{euex}{"52}
\DeclareMathAlphabet{\mathpzc}{OT1}{pzc}{m}{it}

\setenumerate[1]{itemsep=0pt,partopsep=0pt,parsep=\parskip,topsep=1pt}

\begin{document}

\setlength{\abovedisplayshortskip}{5pt}
\setlength{\belowdisplayshortskip}{5pt}
\setlength{\abovedisplayskip}{5pt}
\setlength{\belowdisplayskip}{5pt}

\title{\bf \LARGE{Discounted Expected Utility: A Revealed Preference Analysis}}

\author{Wei Ma\thanks{Corresponding author: Center for Economic Research, Shandong University, Jinan, 250100, China. Email: wei.ma@sdu.edu.cn.}}

\date{}
 \maketitle
\begin{onecolabstract}\noindent
We present a revealed preference characterization of the discounted expected utility model with a concave utility function. The characterization offers a nonparametric test of the model. We apply the test to an experimental data set in the literature and find that the model is almost always rejected even when all payments involved are subject to risk.
\end{onecolabstract}

\section{Introduction}\label{section:Introduction}
In economics, time and risk are a pair of fundamental concepts that are closely intertwined. The canonical theory guiding decision making under risk and over time is the discounted expected utility (DEU) model. Although subjected to various behavioral anomalies, it remains the workhorse in macroeconomics and finance. The empirical content of DEU, however, has not yet been fully appreciated for market behavior. In this paper, we study this question and present a necessary and sufficient revealed-preference condition for a finite set of price-quantity pairs to be consistent with DEU.

We consider a setting with a finite number of discrete time periods and that there are the same number of states of nature in each period. The probability of the occurrence of each state is also the same across the time periods. There is one commodity in each state. Let $x_{st}$ be the consumption of this commodity in state $s$ and period $t$. We assume an individual has a concave (von Neumann-Morgenstern) utility function $u$ and a discount factor $\beta$. Given a consumption stream $\mathbf{x}=(x_{st})$, its DEU is given by
$$U(\mathbf{x})=\sum_{s, t}\beta^t\pi_s u(x_{st}),$$
where $\pi_s$ denotes the probability of state $s$. Let $p_{st}$ be the price of the commodity in state $s$ and period $t$ and $\mathbf{p}=(p_{st})$ a vector of such prices. Suppose that we have a set of observations $(\mathbf{p}^k, \mathbf{x}^k)$, which means the individual chooses the consumption vector $\mathbf{x}^k$ from his budget set when faced with the price vector $\mathbf{p}^k$. Under what condition are these observations consistent with DEU maximization?

The above setting contains two well-known special cases. When there is only one state of nature in each period, it reduces to the exponentially discounted utility (EDU) model; and when there is only one time period, it becomes the subjective expected utility (SEU) model. Revealed-preference characterizations of the two models have been provided by  \cite{Echenique2015} and \cite{Echenique2020}. They both take the form of  the \lq\lq downward sloping demand\rq\rq property subject to some qualifications. For instance, the characterization of EDU is to the effect that if consumption is higher in some later period than that in an earlier period, then the price of the commodity in the later period must be lower. The main result of this paper is that DEU can also be characterized by the \lq\lq downward sloping demand\rq\rq property with proper qualifications. 

In this paper, we distinguish between two scenarios. In the first scenario, the probability distribution on the states of nature is assumed to be known while the discount factor is unknown. This is the setup employed in many experimental investigations of DEU \citep[see, e.g., ][]{Andreoni2012a, Lanier2024}. It also accommodates the EDU model,  for which the state of each period occurs with probability one.  In the second scenario, neither the probability distribution nor the discount factor is known. Compared with the first scenario, the latter one appears more realistic.

Both  \cite{Echenique2015} and \cite{Echenique2020} establish their results using proof by contradiction and the Theorem of the Alternative. Here we demonstrate our result by a somewhat more direct argument. Note that if $\beta$ and $\pi_s$ were known, DEU amounts to objective expected utility (OEU), whose revealed-preference characterization has been given by \cite{Kubler2014} and \citet[][Supplement]{Echenique2015}. Our idea is to explicitly construct $\beta$ and $\pi_s$ from the observations and then appeal to the characterization of OEU. To illustrate, consider the first scenario and see what information on the discount factor can be inferred from the observations. Suppose that in some observation $(\mathbf{p}, \mathbf{x})$, we have $x_{st}>x_{s't'}$ with $t>t'$. The first-order condition for DEU maximization and the concavity of $u$ imply
 $$\frac{u'(x_{st})}{u'(x_{s't'})}=\beta^{t'-t}\cdot \frac{p_{st}/\pi_s}{p_{st'}/\pi_{s'}}<1,\text{ so that }\frac{p_{st}/\pi_s}{p_{st'}/\pi_{s'}}<\beta^{t-t'}.$$
Therefore, the risk-adjusted price ratio on the left-hand side of the second inequality offers a lower bound for the discount factor. By taking $\beta$ to be the supremum of a properly defined set of  risk-adjusted price ratios, we demonstrate that a data set is consistent with DEU if and only if it is consistent with OEU in which the probability of state $s$ in period $t$ is given by $\beta^t\pi_s$. 

Our revealed-preference characterization of DEU can be used as a nonparametric test of the model. An an illustration, we apply the test to an experiment data set of \citeauthor{Andreoni2012a} (2012---henceforth AS). The same test can also be implemented on the other data sets including those of  \cite{Miao2015} and \cite{Lanier2024}. AS's data set contains $80$ subjects who have to make $84$ experimental decisions. In some decisions, all payments are certain while in some others, all payment are risky. We first apply the test to the whole data set of AS and find that only one the subjects conforms with DEU. We then consider two subsets. In the first subset, all payments involved are certain, and in the second, all payment involved are under risk. We find that of $85\%$of  the subjects conforms with DEU (or, more precisely, EDU) in the first case, but only $8.75\%$ of them do in the latter case.

\emph{Related Literature}.---This paper contributes to the literature on the revealed preference analysis of various behavioral models. For the most general utility maximization model, \cite{Samuelson1938} proposes the weak axiom of revealed preference. It is necessary but not sufficient for a data set to be consistent with the utility maximization hypothesis except there are two commodities in total. When the demand is a function, \cite{Houthakker1950} shows that a data set is consistent with the utility maximization hypothesis if and only if it satisfies the strong axiom of revealed preference. And when the demand is not a function but a correspondence, the empirical content of the utility maximization hypothesis is exhausted by the generalized axiom of revealed preference \citep{Afriat1967, Varian1982}.

Beyond the most general utility maximization model, revealed preference theory turns to characterizing some special behavioral models, most notably the models governing decision making under uncertainty and over time. \cite{Varian1983} presents a  characterization of the OEU model when the market is incomplete. For complete markets , \cite{Green1986} characterize the same model by a set of Afriat's inequalities. When there is one commodity in each state, \cite{Kubler2014} show that the model can be characterized by a revealed-preference condition. For SEU models with one commodity in each state, a similar condition is obtained by  \cite{Echenique2015}. Assuming a linear von Neumann-Morgenstern (vNM) utility function, \cite{Chambers2016a} provide a revealed-preference characterization of, among other models,  the variational and maxmin utilities. When the utility function is concave, \cite{Bayer2013} characterize the same two models by a system of nonlinear inequalities. Because of the nonlinearity, the system is very challenging to verify. Using Polymatroid theory, \cite{Demuynck2024a} develop a revealed-preference test for the Choquet expected utility model with ambiguity aversion, which can be implemented very efficiently when the number of states of nature is not too large. When the utility function is not necessarily concave, \cite{Polisson2020} characterize various models under risk and uncertainty, including expected utility, rand-dependent utility, and maxmin utility. 

For models of decision over time, when there is only one observation, \cite{Browning1989} furnishes a revealed-preference characterization of EDU with the discount factor equal to one, and  \cite{Crawford2010}  the habit formation model through a system of Afriat inequalities. When multiple observations are available, \cite{Echenique2020} characterize the EDU models and some of its generalizations including the quasi-hyperbolic discounting model. In the test of the latter model, \citeauthor{Echenique2020} assume the subjects are committed. \cite{Blow2021} and \cite{Echenique2023a} discuss revealed preference analysis of the model without commitment. For DEU, \cite{Lanier2024} conduct an experiment to elicit preferences over state-contingent timed payouts and find strong evidence against one of central properties of DEU---correlation neutrality. Their analysis does not assume concavity of the utility function. In this paper, by imposing the concavity assumption, we provide a full revealed-preference characterization of DEU.

The rest of the paper is organized as follows. Section~\ref{section:The Setup} presents the setup and precisely formulates the problem.  Section~\ref{section:Revealed Preference Characterization} provides a revealed preference characterization of DEU in both scenarios. In Section~\ref{section: Empirical Illustration}, we apply the revealed preference test to AS's experimental data set. Finally, we conclude in Section~\ref{section: Concluding Remarks} with some remarks.

\section{The Setup}\label{section:The Setup}
Suppose that there are a finite number of time periods, indexed by $t=0, 1, \ldots, T$. In each period, there are a finite number of states of nature, indexed by $s=1, \ldots, S$. For notational convenience, let $\mathds{S}=\{1, \ldots, S\}$ and $\mathds{T}=\{0, 1, \ldots, T\}$. In each state, there is a single commodity. Let $p_{st}$ be the price of the commodity in state $s$ and period $t$ and $x_{st}$ its quantity. The probability of each state is independent of the time periods. Let $\pi_s$ be the probability of state $s\in \mathds{S}$. Let $\mathbf{p}=(p_{st})_{s\in \mathds{S}, t\in \mathds{T}}$ and $\mathbf{x}=(x_{st})_{s\in \mathds{S}, t\in \mathds{T}}$. Assume that an individual is a DEU maximizer, whose discount factor is $\beta$ and (vNM) utility function is $u: \mathbb{R}_+\rightarrow  \mathbb{R}$. Then faced with a price vector $\mathbf{p}\in\mathbb{R}_{++}^{S(T+1)}$ and by normalizing his income to be unity, the individual seeks his consumption bundle by solving the problem
\begin{equation}\label{equation: OP}
\begin{cases}
  \max & \sum\limits_{t=0}^{T}\beta^t\sum\limits_{s=1}^{S}\pi_s u(x_{st}) \\
  \mbox{s. t.}, & \mathbf{p}\cdot \mathbf{x}\leq 1.
\end{cases}
\end{equation}

Suppose that we have a set of $K$ observations on the individual's choice behavior: $\mathcal{O}=\{(\mathbf{p}^k, \mathbf{x}^k)\}_{k\in \mathds{K}}$, where $\mathbf{p}^k\in\mathbb{R}_{++}^{S(T+1)}$, $\mathbf{x}^k\in \mathbb{R}_{+}^{S(T+1)}$, and $\mathds{K}=\{1, \ldots, K\}$. This means that for each $k\in \mathds{K}$, the individual chooses consumption bundle $\mathbf{x}^k$ from his budget set when faced with the price vector $\mathbf{p}^k$. Given the data set $\mathcal{O}$, how to determine whether the individual's behavior is consistent with DEU. To be more precise, we formally state the following definition.
\begin{definition}\label{definition: DEU rational}
The data set $\mathcal{O}$ is DEU rational if there exist a discount factor $\beta\in (0,1]$, a probability measure $\pi \in \mathbb{R}_{++}^{S}$ on $\mathds{S}$ and a strictly increasing and concave utility function $u: \mathbb{R}_{+}\rightarrow \mathbb{R}$ such that $ \mathbf{x}^k$ solves Problem~\eqref{equation: OP} when the price vector $\mathbf{p}=\mathbf{p}^k$ for all $k\in \mathds{K}$.
\end{definition}

For ease of exposition, we make the convention that all utility functions in what follows are meant to be strictly increasing and concave on $\mathbb{R}_{+}$. Sometimes to explicitly indicate the fundamentals $(\beta, \pi, u)$ that rationalize $\mathcal{O}$, we refer to a DEU rational data set as $(\beta, \pi, u)$-DEU rational. In this paper, we address the question: What properties must the data set $\mathcal{O}$ satisfy in order to be DEU rational? Noting that when $S=1$, DEU reduces to EDU, and when $T=1$, it reduces to SEU, DEU rational includes as special cases EDU rational \citep{Echenique2020} and SEU rational \citep{Echenique2015}.

\section{Revealed Preference Characterization}\label{section:Revealed Preference Characterization}
In this section, we present a revealed preference characterization of DEU rational data sets. We first study the case where the probability measure $\pi$ on $\mathds{S}$ is known. This setup is employed in many experimental studies of time and risk preferences \citep{Andreoni2012a, Lanier2024}. Then we proceed to examine the more general case where $\pi$ is unknown.

\subsection{The Case with Objective Probabilities}\label{section: The Case with Objective Probabilities}
Suppose the probability of the occurrence of state $s$ is known to be $\pi_s\in \mathbb{R}_{++}$, $s\in \mathds{S}$. Let $\pi=(\pi_1, \ldots, \pi_S)$. We call $\mathcal{O}$ objective DEU rational if  there exist a discount factor $\beta$ and a utility function $u$ such that it is $(\beta, \pi, u)$-DEU rational. Then under what condition is $\mathcal{O}$ objective DEU rational?

To answer this question, inspired by the ideas of \citet{Echenique2015} and \citet{Echenique2020}, we present the following definition.
\begin{definition}\label{definition: k balanced}
A sequence of pairs $\left(x_{s_it_i}^{k_i}, x_{s'_it'_i}^{k'_i}\right)_{i=1}^n$ is balanced if $x_{s_it_i}^{k_i}> x_{s'_it'_i}^{k'_i}$ for all $i=1, \ldots, n$ and if each $k$ appears as $k_i$ (on the left of the pair) the same number of times it appears as $k'_i$ (on the right of the pair).
\end{definition}

Now suppose that $\mathcal{O}$ is $(\beta, \pi, u)$-DEU rational. Assume for a moment that $u$ is differentiable. For each $k\in \mathds{K}$,  the first-order condition for \eqref{equation: OP} is then given by
\begin{equation}\label{equation: FOC}
\beta^t\pi_s u'(x_{st}^k)=\lambda^k p_{st}^k.
\end{equation}
To avoid cumbersome notation, let $\rho_{st}^k=p_{st}^k/\pi_s$; denote a generic sequence of pairs $\left(x_{s_it_i}^{k_i}, x_{s'_it'_i}^{k'_i}\right)_{i=1}^n$ by $\omega$ and let
\begin{equation}\label{equation: T0Po}
\mathpzc{T}_0(\omega)=\sum_{i=1}^{n}(t_i-t'_i)\text{ and }\mathpzc{P}_o(\omega)=\prod_{i=1}^n \frac{\rho_{s_it_i}^{k_i}}{\rho_{s'_it'_i}^{k'_i}},
\end{equation}
where the subscript $o$ in $\mathpzc{P}_o$ is short for \lq\lq objective.\rq\rq Since $u$ is concave, it follows that $x_{s_it_i}^{k_i}> x_{s'_it'_i}^{k'_i}$ implies $u'(x_{s_it_i}^{k_i})\leq u'(x_{s'_it'_i}^{k'_i})$. To facilitate understanding, consider a balanced sequence with two pairs $\omega=\left(x_{s_it_i}^{k_i}, x_{s'_it'_i}^{k'_i}\right)_{i=1}^2$.  By \eqref{equation: FOC} and the balancedness of $\omega$, we have
\begin{equation}\label{equation: necessary condition}
\frac{u'(x_{s_1t_1}^{k_1})}{u'(x_{s'_1t'_1}^{k'_1})}\cdot \frac{u'(x_{s_2t_2}^{k_2})}{u'(x_{s'_2t'_2}^{k'_2})}=\frac{\lambda^{k_1}\rho_{s_1t_1}^{k_1}/\beta^{t_1}}{\lambda^{k'_1}\rho_{s'_1t'_1}^{k'_1}/\beta^{t'_1}}\cdot \frac{\lambda^{k_2}\rho_{s_2t_2}^{k_2}/\beta^{t_2}}{\lambda^{k'_2}\rho_{s'_2t'_2}^{k'_2}/\beta^{t'_2}}=\beta^{-\mathpzc{T}_0(\omega)}\mathpzc{P}_o(\omega)\leq 1.
\end{equation}
Then if $\mathpzc{T}_0(\omega) \geq 0$, we get $\mathpzc{P}_o(\omega)\leq 1$ as $\beta\in (0,1]$. This is easily seen to be true for any balanced sequence $\omega$ with $\mathpzc{T}_0(\omega) \geq 0$. Formally, we state the following definition.
\begin{definition}\label{definition: SAR-ODEU}
The data set $\mathcal{O}$  satisfies the strong axiom of revealed objective discounted expected utility (SAR-ODEU) if $\mathpzc{P}_o(\omega)\leq 1$ for any balanced sequence $\omega$ with $\mathpzc{T}_0(\omega) \geq 0$.
\end{definition}

It turns out that SAR-ODEU is not only necessary but also sufficient for a data set to be objective DEU rational.
\begin{theorem}\label{theorem: SAR-ODEU}
The data set $\mathcal{O}$ is objective DEU rational if and only if it satisfies SAR-ODEU.
\end{theorem}

When $S=1$, we have $\pi_s=1$ and DEU reduces to EDU. In this case, SAR-ODEU is equivalent to the condition of SAR-EDU in \citet{Echenique2020}. They prove that SAR-EDU characterizes EDU by invoking the theorem of the alternative. Here by directly constructing the discount factor $\beta$, we prove Theorem~\ref{theorem: SAR-ODEU} without recourse to that theorem. 

The intuition of the proof runs as follows. From \eqref{equation: necessary condition}, we have that for any balanced sequence $\omega$ with $\mathpzc{T}_0(\omega) \geq 0$,
$$\beta\geq \left(\mathpzc{P}_o(\omega)\right)^{1/\mathpzc{T}_0(\omega)}.$$
In other words, $\left(\mathpzc{P}_o(\omega)\right)^{1/\mathpzc{T}_0(\omega)}$ provides a lower bound for the discount factor. Then take
\begin{equation}\label{equation: beta construction}
\beta=\sup_{\omega}\left(\mathpzc{P}_o(\omega)\right)^{1/\mathpzc{T}_0(\omega)},
\end{equation}
where the supremum is taken over all balanced sequence $\omega$ with $\mathpzc{T}_0(\omega) \geq 0$. With $\beta$ given, we can regard each pair $(s,t)$ as a single state which occurs with probability $\beta^t\pi_s$ and, therefore, regard Problem~\eqref{equation: OP} as an objective expected utility (OEU) maximization problem. Let $\bar{\rho}_{st}^k=\rho_{st}^k/\beta^t$. By Proposition~S.9 of \citet{Echenique2015}, the data set $\mathcal{O}$ is OEU rational if and only if for any balanced sequence $\omega=\left(x_{s_it_i}^{k_i}, x_{s'_it'_i}^{k'_i}\right)_{i=1}^n$,
\begin{equation}\label{equation: OEU rational}
\prod_{i=1}^n \frac{\bar{\rho}_{s_it_i}^{k_i}}{\bar{\rho}_{s'_it'_i}^{k'_i}}\leq 1.
\end{equation}
But \eqref{equation: OEU rational} is an immediate consequence of \eqref{equation: beta construction} by noting that $\prod_{i=1}^n (\bar{\rho}_{s_it_i}^{k_i}/\bar{\rho}_{s'_it'_i}^{k'_i})=\beta^{-\mathpzc{T}_0(\omega)}\mathpzc{P}_o(\omega)$.

\textbf{Remark}. Here the prior $\pi$ on $\mathds{S}$ is assumed fixed over time and across observations. \cite{Kubler2014} allow $\pi$ to vary with $k\in \mathds{K}$. As noted by \citet[][Section S.6]{Echenique2015}, their Proposition~S.9 can be generalized to this variable-prior setup. With the aid of this result, Theorem~\ref{theorem: SAR-ODEU} can be extended to the setup with priors varying both over time and across observations. More specifically, let $\pi_{st}^k$ be the probability of state $s$ in period $t$ for observation $k$, in which it is not required that $\pi_{st}^k=\pi_{st'}^k$ for $t\neq t'$. Theorem~\ref{theorem: SAR-ODEU} still holds in the case by redefining $\rho_{st}^k=p_{st}^k/\pi_{st}^k$. This general setup will be useful in Section~\ref{section: Empirical Illustration}.

\subsection{The Case with Subjective Probabilities}
In this section, we consider the case where the probability measure $\pi$ on $\mathds{S}$ is unknown. To this end, we need to strengthen the concept of balancedness of a sequence.

Let $\mathds{S}_j=\{j+1, \ldots, S\}$, $j=0, 1, \ldots, S-1$, and we make the convention that $\mathds{S}_S=\varnothing$. 
\begin{definition}\label{definition: s balanced}
For any $j=0, 1, \ldots, S$, a sequence of pairs $\left(x_{s_it_i}^{k_i}, x_{s'_it'_i}^{k'_i}\right)_{i=1}^n$ is $\mathds{S}_j$-balanced if  it is balanced and each $s\in \mathds{S}_j$ appears as $s_i$ (on the left of the pair) the same number of times it appears as $s'_i$ (on the right of the pair).
\end{definition}

A remark is in order here. A balanced sequence, as defined in Definition~\ref{definition: k balanced}, refers to the number of times each $k\in \mathds{K}$ appears on the left and right of a sequence. Besides requiring this \lq\lq$k$-balancedness,\rq\rq an  $\mathds{S}_j$-balanced sequence also imposes a restriction on the number of times each $s\in \mathds{S}_j$ appears on both sides of a sequence. In particular, because $\mathds{S}_S=\varnothing$, an $\mathds{S}_S$-balanced sequence is just a balanced sequence.

Suppose that $\mathcal{O}$ is $(\beta, \pi, u)$-DEU rational. Given a sequence $\omega=\left(x_{s_it_i}^{k_i}, x_{s'_it'_i}^{k'_i}\right)_{i=1}^n$, let
$$\mathpzc{P}_0(\omega)=\prod_{i=1}^n \frac{p_{s_it_i}^{k_i}}{p_{s'_it'_i}^{k'_i}},$$
where, in contrast to $\mathpzc{P}_o$ defined in \eqref{equation: T0Po}, the subscript $0$ in $\mathpzc{P}_0$ here is the number \lq\lq zero\rq\rq, not the English letter \lq\lq o.\rq\rq If $\omega$ is $\mathds{S}_0$-balanced, it follows from the first-order condition~\eqref{equation: FOC} that
\begin{equation}\label{equation: P0 FOC}
\prod_{i=1}^{n}\frac{u'(x_{s_it_i}^{k_i})}{u'(x_{s'_it'_i}^{k'_i})}=\beta^{-\mathpzc{T}_0(\omega)}\mathpzc{P}_0(\omega)\leq 1.
\end{equation}
Again if $\mathpzc{T}_0(\omega) \geq 0$, $\mathpzc{P}_0(\omega)\leq 1$ as $\beta\in (0,1]$.  We formally state this as the following definition.
\begin{definition}\label{definition: SAR-DEU}
The data set $\mathcal{O}$  satisfies the strong axiom of revealed discounted expected utility (SAR-DEU) if $\mathpzc{P}_0(\omega)\leq 1$ for any $\mathds{S}_0$-balanced sequence $\omega$ with $\mathpzc{T}_0(\omega) \geq 0$.
\end{definition}

As in the previous subsection, SAR-DEU turns out to be both necessary and sufficient for a data set to be DEU rational.
\begin{theorem}\label{theorem: SAR-DEU}
The data set $\mathcal{O}$ is DEU rational if and only if it satisfies SAR-DEU.
\end{theorem}

When $T=0$, this theorem is identical with Theorem~1 of \cite{Echenique2015}. The intuition of the theorem goes as follows. By considering $\mathds{S}_0$-balanced sequences, we get \eqref{equation: P0 FOC}, which gives
$$\beta\geq \left(\mathpzc{P}_0(\omega)\right)^{1/\mathpzc{T}_0(\omega)}.$$
Therefore, $\left(\mathpzc{P}_0(\omega)\right)^{1/\mathpzc{T}_0(\omega)}$ furnishes a lower bound for $\beta$. Take $\beta=\sup_{\omega}\left(\mathpzc{P}_0(\omega)\right)^{1/\mathpzc{T}_0(\omega)}$, where the supremum is take over all $\mathds{S}_0$-balanced sequences. To determine $\pi_1$, consider $\mathds{S}_1$-balanced sequences and let $p_{st}^{k1}=p_{st}^k/\beta^t$. Given a sequence $\omega=\left(x_{s_it_i}^{k_i}, x_{s'_it'_i}^{k'_i}\right)_{i=1}^n$, let 
\begin{equation}\label{equation: T1P1}
\begin{aligned}
& \mathpzc{T}_1(\omega)=\sharp s_1\text{ appears on the left of }\omega-\sharp s_1\text{ appears on the right of }\omega,\\
  &\mathpzc{P}_1(\omega)=\prod_{i=1}^n \frac{p_{s_it_i}^{k_i, 1}}{p_{s'_it'_i}^{k'_i, 1}}.
\end{aligned}
\end{equation}
If $\omega$ is $\mathds{S}_1$-balanced, we have
$$\prod_{i=1}^{n}\frac{u'(x_{s_it_i}^{k_i})}{u'(x_{s'_it'_i}^{k'_i})}=\pi_1^{-\mathpzc{T}_1(\omega)}\mathpzc{P}_1(\omega)\leq 1,$$
so that $\pi_1\geq \left(\mathpzc{P}_1(\omega)\right)^{1/\mathpzc{T}_1(\omega)}$. Define $\pi_1=\sup_{\omega}\left(\mathpzc{P}_0(\omega)\right)^{1/\mathpzc{T}_0(\omega)}$, where the supremum is take over all $\mathds{S}_1$-balanced sequences. Repeating this process, we will finally obtain $\pi_S$. Given $\beta$ and $(\pi_s)_{s\in \mathds{S}}$, Problem~\eqref{equation: OP} reduces to an OEU maximization problem. Therefore, we can apply Proposition~S.9 of \citet{Echenique2015} to show the DEU rationality of $\mathcal{O}$.

\section{Empirical Illustration}\label{section: Empirical Illustration}
In this section, we apply our theoretical results in the preceding section to the experimental data of AS. We first briefly describe AS's experiment and then present the test result.

AS's experiment employs a method called Convex Time Budget (CTB) introduced in an earlier paper \citep{Andreoni2012}. Specifically, each subject is endowed with $100$ experimental tokens and asked to allocate them between a sooner date ($t$ days from the experiment date) and a later date ($t+d$ days from the experiment date). Tokens allocated to the later date always have a value of $a_{t+d}=\$ 0.2$, while the value $a_t$ of those allocated to the sooner date varies from $\$0.2$ to $\$0.14$ per token. The payments scheduled for both dates are subject to risk. Let $\pi_1$ and $\pi_2$ be the respective probabilities that the payment will be made for the sooner and the later dates. In all the CTB decisions, we have $t=7$, $d\in \{28, 56\}$, $(\pi_1, \pi_2)\in \{(1, 1), (0.5, 0.5), (1, 0.8), (0.5, 0.4), (0.8, 1), (0.4, 0.5)\}$, and $a_t\in \{0.2, 0.19, 0.18, 0.17, 0.16, 0.15, 0.14\}$. As such, each CTB decision of a subject can be described by a tuple $\gamma=(d, a_t, c_t, \pi_1, \pi_2)$, where $c_t$ is the number of tokens allocated to the sooner date.

AS find that DEU is rejected if some payments are made with certainty, but fits very well with the data if all payments involved are risky. Let us verify this finding using the theoretical results in the preceding section. To this end, let $\Gamma$ be the set of all tuples $\gamma$ observed in AS's experiment,  $\Gamma_1$ the set of tuples with certain payments and $\Gamma_2$ all tuples with risky payments, i.e.
\begin{align*}
 \Gamma_1&=\{(d, a_t, c_t, \pi_1, \pi_2)\in \Gamma: (\pi_1, \pi_2)=(1, 1)\}, \\
  \Gamma_2&=\{(d, a_t, c_t, \pi_1, \pi_2)\in \Gamma: (\pi_1, \pi_2)\in \{(0.5, 0.5), (0.5, 0.4), (0.4, 0.5)\}\}.
\end{align*}
We verify whether each data set satisfies SAR-ODEU. The details of the implementation is explained in Appendix~\ref{section: Implementation of the Revealed Preference Test }. The test result is presented in Tab.~\ref{table: pass rate}. Form it, we can see that out of  the $80$ subjects in AS's experiment, only one of them conform to DEU when all observations (i.e. $\Gamma$) are considered, and this number just slightly increases to $7$ when considering the risky observations (i.e. $\Gamma_2$). This is somewhat at variance with the findings of AS. But the number increases dramatically to $68$ if only observations with certain payments (i.e. $\Gamma_1$) are considered. This means $85\%$ of the subjects behave in line with EDU. 
\begin{table}
  \centering  \caption{Pass Rate}\label{table: pass rate}
\begin{tabular}{ccc}
  \hline\hline
  Data set & Number of choices & Pass rate \\\hline
  $\Gamma$ & 84 & 0.0125\\
  $\Gamma_1$ & 20 & 0.85 \\
  $\Gamma_2$ & 30 & 0.0875 \\
  \hline
\end{tabular}
\end{table}

\section{Concluding Remarks}\label{section: Concluding Remarks}

This paper provides a revealed-preference characterization of  DEU models with a concave vNM utility function. The analysis covers two cases: one in which state probabilities are known and only the discount factor is to be recovered, and a more general case in which both the discount factor and the state probabilities are unknown. In both settings, we offer a nonparametric test of DEU.

We apply the test to the experimental data of AS. The results indicate that DEU is rejected for nearly all subjects in the full data set, and its performance only shows a modest increase when restricted to observations with risky payments. This finding contrasts with the view that DEU fits well once risk is introduced into all payments.

For future research, note that the information structure of the DEU model studied in this paper is described by a constant filtration; that is, the information algebra is the same in every period, so no new information arrives over time. This assumption is well suited to many experimental investigations of DEU models. In many real-world environments, however, information accumulates over time and learning is possible. A natural information structure for capturing such environments is an increasing filtration \citep{Kochov2015}. How DEU can be tested under such an information structure? If the state probabilities in each period are known, Theorem~\ref{theorem: SAR-ODEU} continues to hold verbatim. When the state probabilities are unknown, however, Theorem~\ref{theorem: SAR-DEU} does not directly carry over. Therefore, developing tests of DEU under the latter situation constitutes an important avenue for future research.

\appendix
\begin{appendices}
\numberwithin{equation}{section}
\section{Proofs}

\subsection{Proof of Theorem~\ref{theorem: SAR-ODEU}}\label{section: Proof of Theorem SAR-ODEU}

Before proceeding, we reproduce a result of \citet{Echenique2015}. Take $\mathds{S}\times \mathds{T}$ as a set of states. Suppose it is known that state $(s,t)$ occurs with probability $\pi_{st}>0$. The data set $\mathcal{O}$ is objective expected utility (OEU) rational if there exists a  utility function $u: \mathbb{R}_+\rightarrow \mathbb{R}$ such that $\mathbf{p}^k\cdot \mathbf{y}\leq \mathbf{p}^k\cdot \mathbf{x}^k$ implies $\sum_{s, t}\pi_{st}u(y_{st})\leq \sum_{s, t}\pi_{st}u(x_{st}^k)$. The following result characterizes OEU rational data sets.
\begin{lemma}[Proposition~S.9 of \citet{Echenique2015}]\label{lemma: OEU rational}
The data set $\mathcal{O}$ is OEU rational if and only if for every balanced sequence $\left(x_{s_it_i}^{k_i}, x_{s'_it'_i}^{k'_i}\right)_{i=1}^n$, it holds that 
$\prod_{i=1}^n \frac{p_{s_it_i}^{k_i}/\pi_{s_it_i}}{ p_{s'_it'_i}^{k'_i}/\pi_{s'_it'_i}}\leq 1.$
\end{lemma}

\textbf{Remark}. Lemma~\ref{lemma: OEU rational} holds valid without requiring $\sum_{s, t}\pi_{st}=1$. See also \citet[][p. 3462]{Kubler2014}.

We are now equipped to establish the theorem. Let us start with the necessity part. Suppose that $\mathcal{O}$ is objective DEU rational. Then there exist a discount factor $\beta$ and  a  utility function $u: \mathbb{R}_+\rightarrow \mathbb{R}$ such that $\mathcal{O}$ is $(\beta, \pi, u)$-DEU rational. Take $(s,t)$ as a single state occurring with probability $\pi_{st}=\pi_s\beta^t$ for all $s\in \mathds{S}$ and $t\in \mathds{T}$, so that $\mathcal{O}$ is OEU rational.  For any balanced sequence $\omega=\left(x_{s_it_i}^{k_i}, x_{s'_it'_i}^{k'_i}\right)_{i=1}^n$, we have, by Lemma~\ref{lemma: OEU rational},
$$\prod_{i=1}^n \frac{p_{s_it_i}^{k_i}/\pi_{s_it_i}}{ p_{s'_it'_i}^{k'_i}/\pi_{s'_it'_i}}= \beta^{-\mathpzc{T}_0(\omega)}\mathpzc{P}_o(\omega)\leq 1.$$
If $\mathpzc{T}_0(\omega)\geq 0$, we get $\mathpzc{P}_o(\omega)\leq 1$ as $\beta\in (0,1]$, hence $\mathcal{O}$ satisfies SAR-ODEU.

We proceed to the sufficiency part. Suppose that $\mathcal{O}$ satisfies SAR-ODEU. Let $\Omega$ be the set of all balanced sequences and
\begin{align*}
\Omega_1&=\{\omega\in \Omega: \mathpzc{T}_0(\omega)<0, \mathpzc{P}_o(\omega)>1\},\\
\Omega_2&=\{\omega\in \Omega: \mathpzc{T}_0(\omega)>0\}.
\end{align*}
With these notations, we define
\begin{equation}\label{equation: discount factor}
\beta=
  \begin{cases}
    1, & \mbox{if } \Omega_1= \varnothing\text{ and } \Omega_2= \varnothing \\
   \inf_{\omega\in \Omega_1}[\mathpzc{P}_o(\omega)]^{1/\mathpzc{T}_0(\omega)}, & \mbox{if }\Omega_1\neq \varnothing\text{ and } \Omega_2= \varnothing \\
    \sup_{\omega\in \Omega_2}[\mathpzc{P}_o(\omega)]^{1/\mathpzc{T}_0(\omega)}, & \mbox{if } \Omega_2\neq \varnothing. 
  \end{cases}
\end{equation}
It is obvious that $\beta\in [0,1]$. To show $\beta>0$, we introduce a new concept. A balanced sequence is primal if it does not contain any proper subsequence which is also balanced. More precisely, a balanced sequence $\omega=\left(x_{s_it_i}^{k_i}, x_{s'_it'_i}^{k'_i}\right)_{i=1}^n$ is primal if there does not exist any proper subset $I$ of $\{1, 2, \ldots, n\}$ such that the sequence $\omega_1=\left(x_{s_it_i}^{k_i}, x_{s'_it'_i}^{k'_i}\right)_{i\in I}$  is also balanced. Let $\omega_2=\left(x_{s_it_i}^{k_i}, x_{s'_it'_i}^{k'_i}\right)_{i\in \{0, \ldots, n\}\backslash I}$. Note that the order each pair $(x_{s_it_i}^{k_i}, x_{s'_it'_i}^{k'_i})$ appears in the sequence $\omega$ has no bite, so in the following we shall write $\omega=(\omega_1, \omega_2)$ and $\omega_2=\omega\backslash \omega_1$. Then if $\omega$ is not primal, there exists a proper subsequence $\omega_1$ of $\omega$ which is also balanced. And, as a result, the sequence $\omega_2=\omega\backslash \omega_1$  is non-empty and balanced. Let $\Omega_p$ be the set of primal sequences in $\Omega$.
\begin{lemma}\label{lemma: Omega primal}
$\Omega_p$ is finite.
\end{lemma}
\begin{proof}
As in the proof of Proposition~2 of \cite{Echenique2015}, let
$$\sum=\left\{(k, s, t, k', s', t')\in \mathds{K}\times \mathds{S}\times\mathds{T}\times \mathds{K}\times \mathds{S}\times\mathds{T}: x_{st}^k>x_{s't'}^{k'}\right\}.$$
Denote the cardinality of $\sum$ by $N$. Construct a $K\times N$ matrix $G$, with each element in $\mathds{K}$  corresponding to a row and each element in $\sum$ corresponding to a column. The entry in row $\hat{k}\in \mathds{K}$ and column $(k, s, t, k', s', t')\in \sum$ of $G$ is $1$ if $\hat{k}=k$, $-1$ if $\hat{k}=k'$, and $0$ otherwise. Let $\mathbb{Z}_+$ be the set of non-negative integers. For notational convenience, let $\phi$ be a bijection from $\mathds{N}$ to $\sum$, where $\mathds{N}=\{1, \ldots, N\}$. Define a mapping $\psi: \Omega\rightarrow \mathbb{Z}_+^N$ such that for each $\omega\in \Omega$, the $i$-th component of $\psi(\omega)$ is equal to the number of times the pair $(x_{st}^k, x_{s't'}^{k'})$ appears in $\omega$, where $(k, s, t, k', s', t')= \phi(i)\in \sum$. Let $\mathcal{V}=\{v\in \mathbb{Z}_+^N: Gv=0\}$. Then since $\Omega$ is the set of balanced sequences, $\psi$  is a bijection between $\Omega$ and $\mathcal{V}$. A vector $v\in \mathcal{V}$ is a minimal element of $\mathcal{V}$ if there is no other nonzero element $v'\in \mathcal{V}$ such that $v'\leq v$ and $v'\neq v$. Let $\mathcal{V}_m$ be the set of minimal elements of  $\mathcal{V}$. Then it is not hard to see $\mathcal{V}_m=\psi(\Omega_p)$. By Dickson's lemma  \citep[see][Theorem~5.1, p. 48]{Rosales}, $\mathcal{V}_m$ is finite and, therefore, so is $\Omega_p$.
\end{proof}

\begin{lemma}\label{lemma: beta 0 1}
$\beta\in (0,1]$.
\end{lemma}
\begin{proof}
This is obviously true if $\Omega_1= \Omega_2= \varnothing$. If $\Omega_2\neq \varnothing$,  the lemma follows at once from SAR-ODEU. It remains to consider the case where $\Omega_1\neq \varnothing$  and $\Omega_2= \varnothing$. In this case, we show $\beta=\min_{\omega\in \Omega_1 \cap \Omega_p} [\mathpzc{P}_o(\omega)]^{1/\mathpzc{T}_0(\omega)}$. For notational convenience, let $F(\omega)=[\mathpzc{P}_o(\omega)]^{1/\mathpzc{T}_0(\omega)}$. It suffices to show that for any $\omega\in \Omega_1 \backslash \Omega_p$, there exists a $\omega_p\in \Omega_p$ such that $F(\omega)\geq F(\omega_p)$. Fix $\omega\in \Omega_1 \backslash \Omega_p$. Let $\omega_1$ be a balanced subsequence of $\omega$ and $\omega_2=\omega\backslash \omega_1$. We have therefore
\begin{equation}\label{equation: omega 12}
  \mathpzc{P}_o(\omega)=\mathpzc{P}_o(\omega_1)\cdot \mathpzc{P}_o(\omega_2)\text{ and }\mathpzc{T}_0(\omega)=\mathpzc{T}_0(\omega_1)+\mathpzc{T}_0(\omega_2).
\end{equation}
Since $\mathpzc{T}_0(\omega)<0$, at least one of the two inequalities must be true: $\mathpzc{T}_0(\omega_1)<0$ or  $\mathpzc{T}_0(\omega_2)<0$. Assume for definiteness that $\mathpzc{T}_0(\omega_2)<0$. As $\Omega_2= \varnothing$, we have $\mathpzc{T}_0(\omega_1)\leq0$. If $\mathpzc{T}_0(\omega_1)=0$, then $\mathpzc{T}_0(\omega)=\mathpzc{T}_0(\omega_2)$ and by SAR-ODEU, $\mathpzc{P}_o(\omega_1)\leq 1$, so that $\mathpzc{P}_o(\omega_2)\geq \mathpzc{P}_o(\omega)$. This implies $\omega_2\in \Omega_1$. Since $\mathpzc{T}_0(\omega)=\mathpzc{T}_0(\omega_2)<0$, it follows that $F(\omega)\geq F(\omega_2)$. On the other hand, if $\mathpzc{T}_0(\omega_1)<0$, we have
\begin{equation}\label{equation: P0T0}
\frac{\ln \mathpzc{P}_o(\omega)}{\mathpzc{T}_0(\omega)}=\frac{\ln \mathpzc{P}_o(\omega_1)+\ln \mathpzc{P}_o(\omega_2)}{\mathpzc{T}_0(\omega_1)+\mathpzc{T}_0(\omega_2)}\geq \min\left\{\frac{\ln \mathpzc{P}_o(\omega_1)}{\mathpzc{T}_0(\omega_1)}, \frac{\ln \mathpzc{P}_o(\omega_2)}{\mathpzc{T}_0(\omega_2)}\right\}.
\end{equation}
Since $\mathpzc{P}_o(\omega)>1$, it follows from \eqref{equation: omega 12} that either $\mathpzc{P}_o(\omega_1)>1$  or $\mathpzc{P}_o(\omega_2)>1$, hence at least one of $\omega_1$ and $\omega_2$ must belong to $\Omega_1$. Assume for definiteness that $\omega_2\in \Omega_1$. If also $\omega_1\in \Omega_1$, then by \eqref{equation: P0T0}, there exists a $\omega'\in \{\omega_1, \omega_2\}$ such that 
$$\frac{\ln \mathpzc{P}_o(\omega)}{\mathpzc{T}_0(\omega)}\geq\frac{\ln \mathpzc{P}_o(\omega')}{\mathpzc{T}_0(\omega')}, \text{ or equivalently, } F(\omega)\geq F(\omega').$$
If $\omega_1\notin \Omega_1$, then $\mathpzc{P}_o(\omega_1)\leq 1$, so that $\mathpzc{P}_o(\omega_2)> 1$. This means $\omega_2\in \Omega_1$ and 
$\ln \mathpzc{P}_o(\omega_1)/\mathpzc{T}_0(\omega_1)> \ln \mathpzc{P}_o(\omega_2)/\mathpzc{T}_0(\omega_2)$. Consequently, by \eqref{equation: P0T0}, $F(\omega)\geq F(\omega_2).$

To summarize, for any $\omega\in \Omega_1 \backslash \Omega_p$, we have obtained  a proper subsequence $\omega''\in \Omega_1$ of $\omega$ such that $F(\omega)\geq F(\omega'').$ Since $\omega$  is a finite set, by iterating the above procedure, we will finally arrive at a subsequence $\hat{\omega}\in \Omega_p$ of $\omega$ such that $F(\omega)\geq F(\hat{\omega}).$
\end{proof}

\begin{lemma}\label{lemma: sk beta ds}
For every $\omega\in \Omega$, we have
\begin{equation}\label{equation: sk beta ds}
  \beta^{-\mathpzc{T}_0(\omega)}\mathpzc{P}_o(\omega)\leq 1.
\end{equation}
\end{lemma}
\begin{proof}
Take any $\omega\in \Omega$. If $\mathpzc{T}_0(\omega)=0$, we have $\mathpzc{P}_o(\omega)\leq 1$ by SAR-ODEU and hence \eqref{equation: sk beta ds}. If $\mathpzc{T}_0(\omega)>0$, then $\omega\in \Omega_2$ and so by \eqref{equation: discount factor}, $\beta= \sup_{\omega'\in \Omega_2}[\mathpzc{P}_o(\omega')]^{1/\mathpzc{T}_0(\omega')}$. This implies $\beta\geq [\mathpzc{P}_o(\omega)]^{1/\mathpzc{T}_0(\omega)}$, from which \eqref{equation: sk beta ds} follows immediately.

Finally, consider the case of $\mathpzc{T}_0(\omega)<0$. If $\mathpzc{P}_o(\omega)\leq 1$,  in view of $\beta\in (0,1]$, \eqref{equation: sk beta ds} holds obviously. Now suppose $\mathpzc{P}_o(\omega)> 1$, so that $\Omega_1\neq \varnothing$. If  $\Omega_2= \varnothing$, then $\beta=\inf_{\omega'\in \Omega_1}[\mathpzc{P}_o(\omega')]^{1/\mathpzc{T}_0(\omega')}$. Therefore, $\beta\leq [\mathpzc{P}_o(\omega)]^{1/\mathpzc{T}_0(\omega)}$. As $\mathpzc{T}_0(\omega)<0$, we get $\beta^{\mathpzc{T}_0(\omega)}\geq \mathpzc{P}_o(\omega)$, and hence \eqref{equation: sk beta ds}. If  $\Omega_2\neq \varnothing$, take any $\omega_0\in \Omega_2$. Consider the sequence $\bar{\omega}=(\omega_0,\ldots, \omega_0, \omega,\ldots, \omega)$ in which the number of $\omega_0$ is $-\mathpzc{T}_0(\omega)$ and the number of $\omega$ is $\mathpzc{T}_0(\omega_0)$. Since $\omega_0$ and $\omega$ are both balanced, so is $\bar{\omega}$. Noting that $\mathpzc{T}_0(\bar{\omega})=0$, we have, by SAR-ODEU, that $\mathpzc{P}_o(\bar{\omega})\leq 1$ or, equivalently, $(\mathpzc{P}_o(\omega_0))^{-\mathpzc{T}_0(\omega)}\cdot (\mathpzc{P}_o(\omega))^{\mathpzc{T}_0(\omega_0)}\leq 1$. Taking the $\mathpzc{T}_0(\omega_0)$-th root of both sides, we get
$\left[(\mathpzc{P}_o(\omega_0))^{1/\mathpzc{T}_0(\omega_0)}\right]^{-\mathpzc{T}_0(\omega)}\cdot \mathpzc{P}_o(\omega)\leq 1.$
Since this is true for any $\omega_0\in \Omega_2$, it follows that 
$$\left[\sup_{\omega_0\in \Omega_2}(\mathpzc{P}_o(\omega_0))^{1/\mathpzc{T}_0(\omega_0)}\right]^{-\mathpzc{T}_0(\omega)}\cdot \mathpzc{P}_o(\omega)\leq 1,$$
and hence \eqref{equation: sk beta ds}. This proves the lemma.
\end{proof}

As argued in proof of the necessity part, regard each pair $(s,t)$ as a single state which occurs with probability $\beta^t\pi_s$. Let $\bar{\rho}_{st}^k=\rho_{st}^k/\beta^t$. For any sequence $\left(x_{s_it_i}^{k_i}, x_{s'_it'_i}^{k'_i}\right)_{i=1}^n\in \Omega$, we have, by Lemma~\ref{lemma: sk beta ds},
\begin{equation}\label{equation: OEU DEU rational}
\prod_{i=1}^n \frac{\bar{\rho}_{s_it_i}^{k_i}}{\bar{\rho}_{s'_it'_i}^{k'_i}}=\beta^{-\mathpzc{T}_0(\omega)}\mathpzc{P}_o(\omega)\leq 1.
\end{equation}
By Lemma~\ref{lemma: OEU rational}, $\mathcal{O}$ is OEU rational or, equivalently, $(\beta, \pi, u)$-DEU rational.

\subsection{Proof of Theorem~\ref{theorem: SAR-DEU}}
The proof of the necessity part is the same as in Theorem~\ref{theorem: SAR-ODEU}, and so is omitted here. We concentrate on the proof of the sufficiency part.

Suppose that $\mathcal{O}$ satisfies SAR-DEU. We construct $\beta$ and the distribution $\pi$ on $\mathds{S}$ explicitly. The idea is essentially the same as in the proof of Theorem~\ref{theorem: SAR-ODEU}. Let $\Omega_0$ be the set of all $\mathds{S}_0$-balanced sequences and
\begin{align*}
\Omega_{01}&=\{\omega\in \Omega_0: \mathpzc{T}_0(\omega)<0, \mathpzc{P}_0(\omega)>1\},\\
\Omega_{02}&=\{\omega\in \Omega_0: \mathpzc{T}_0(\omega)>0\}.
\end{align*}
Define
\begin{equation}\label{equation: discount factor DEU}
\beta=
  \begin{cases}
    1, & \mbox{if } \Omega_{01}= \varnothing\text{ and } \Omega_{02}= \varnothing \\
   \inf_{\omega\in \Omega_{01}}[\mathpzc{P}_0(\omega)]^{1/\mathpzc{T}_0(\omega)}, & \mbox{if }\Omega_{01}\neq \varnothing\text{ and } \Omega_{02}= \varnothing \\
    \sup_{\omega\in \Omega_{02}}[\mathpzc{P}_0(\omega)]^{1/\mathpzc{T}_0(\omega)}, & \mbox{if } \Omega_{02}\neq \varnothing. 
  \end{cases}
\end{equation}
Letting $p_{st}^k$ taking the place of $\rho_{st}^k$, $\Omega_0$ taking the place of $\Omega$,  and following exactly the argument of Theorem~\ref{theorem: SAR-ODEU}, we will obtain $ \beta\in (0,1]$ and
\begin{equation}\label{equation: Omega0beta}
 \beta^{-\mathpzc{T}_0(\omega)}\mathpzc{P}_0(\omega)\leq 1 \text{ for all }\omega\in \Omega_0.
\end{equation}

We proceed to determine the distribution $\pi$. Let $\Omega_1$ be the set of $\mathds{S}_1$-balanced sequences. Given a sequence $\omega=\left(x_{s_it_i}^{k_i}, x_{s'_it'_i}^{k'_i}\right)_{i=1}^n$, let 
\begin{equation}\label{equation: Ts}
\mathpzc{T}_s(\omega)=\sharp s\text{ appears on the left of }\omega-\sharp s\text{ appears on the right of }\omega, s\in \mathds{S}.
\end{equation}
Let $p_{st}^{k, 1}=p_{st}^k/\beta^t$ and recall the definition of $\mathpzc{P}_1(\omega)$ from \eqref{equation: T1P1}. For any $\omega\in \Omega_0$, we have, by \eqref{equation: Omega0beta},
\begin{equation}\label{equation: Omega-pstk1}
\mathpzc{P}_1(\omega)=\beta^{-\mathpzc{T}_0(\omega)}\mathpzc{P}_0(\omega) \leq 1.
\end{equation}
Note that in the case of determining $\beta$, we have the SAR-DEU condition, but a similar condition is unavailable for states of nature. For this reason, we take a slightly different approach to determining $\pi$.  Let
\begin{align*}
\Omega_{11}&=\{\omega\in \Omega_1: \mathpzc{T}_1(\omega)<0\},\\
\Omega_{12}&=\{\omega\in \Omega_1: \mathpzc{T}_1(\omega)>0\}.
\end{align*}
Define
\begin{equation}\label{equation: pi1 DEU}
\pi_1=
  \begin{cases}
    1, & \mbox{if } \Omega_{11}= \varnothing\text{ and } \Omega_{12}= \varnothing \\
   \inf_{\omega\in \Omega_{11}}[\mathpzc{P}_1(\omega)]^{1/\mathpzc{T}_1(\omega)}, & \mbox{if }\Omega_{11}\neq \varnothing\text{ and } \Omega_{12}= \varnothing \\
    \sup_{\omega\in \Omega_{12}}[\mathpzc{P}_1(\omega)]^{1/\mathpzc{T}_1(\omega)}, & \mbox{if } \Omega_{12}\neq \varnothing. 
  \end{cases}
\end{equation}
\begin{lemma}\label{lemma: pi1 range}
$\pi_1\in (0, \infty)$.
\end{lemma}
\begin{proof}
The lemma is trivially true if $\Omega_{11}=  \Omega_{12}= \varnothing$. Henceforth, assume $\Omega_{11}\neq \varnothing$ or $\Omega_{12}\neq \varnothing$. The proof is quite similar to that of Lemma~\ref{lemma: beta 0 1}. But for the sake of completeness, let us spell out the details.

First, we generalize the notion of primal sequences. A sequence $\omega \in \Omega_1$ is primal if it does not contain a proper subsequence which is also in $\Omega_1$. Let $\Omega_{1p}$ be the set of primal sequences in $\Omega_1$. Following the proof of Lemma~\ref{lemma: Omega primal}, we have that $\Omega_{1p}$ is finite. It suffices to show
\begin{equation*}\label{equation: pi1 DEU Omegap}
\pi_1=
  \begin{cases}
   \inf_{\omega\in \Omega_{11}\cap \Omega_{1p}}[\mathpzc{P}_1(\omega)]^{1/\mathpzc{T}_1(\omega)}, & \mbox{if }\Omega_{11}\neq \varnothing\text{ and } \Omega_{12}= \varnothing \\
    \sup_{\omega\in \Omega_{12}\cap \Omega_{1p}}[\mathpzc{P}_1(\omega)]^{1/\mathpzc{T}_1(\omega)}, & \mbox{if } \Omega_{12}\neq \varnothing. 
  \end{cases}
\end{equation*}

Let us prove the case $\Omega_{12}\neq \varnothing$; the other case can be treated likewise. For convenience of notation, let $F_1(\omega)=[\mathpzc{P}_1(\omega)]^{1/\mathpzc{T}_1(\omega)}$. As in the proof of Lemma~\ref{lemma: beta 0 1}, it suffices to show that for any $\omega \in \Omega_{12}$, there exists a sequence $\omega^* \in \Omega_{12}\cap \Omega_{1p}$ such that $F_1(\omega^*)\geq F_1(\omega)$. 

Fix $\omega \in \Omega_{12}\backslash \Omega_{1p}$. Then there exists  a proper subsequence $\omega_1$ of $\omega$ which is also in $\Omega_1$, so that $\omega_2=\omega\backslash \omega_1\in \Omega_1$. Note that 
\begin{equation}\label{equation: P1omega12}
\mathpzc{P}_1(\omega)=\mathpzc{P}_1(\omega_1)\cdot \mathpzc{P}_1(\omega_2)\text{ and }\mathpzc{T}_1(\omega)=\mathpzc{T}_1(\omega_1)+ \mathpzc{T}_1(\omega_2).
\end{equation}
Since $\mathpzc{T}_1(\omega)>0$, at least one of $\mathpzc{T}_1(\omega_1)$ and $\mathpzc{T}_1(\omega_2)$ must be strictly positive. Assume for definiteness that $\mathpzc{T}_1(\omega_2)>0$. If $\mathpzc{T}_1(\omega_1)>0$, we have, by \eqref{equation: P1omega12},
\begin{equation}\label{equation: P1T1}
\frac{\ln \mathpzc{P}_1(\omega)}{\mathpzc{T}_1(\omega)}=\frac{\ln \mathpzc{P}_1(\omega_1)+\ln \mathpzc{P}_1(\omega_2)}{\mathpzc{T}_1(\omega_1)+\mathpzc{T}_1(\omega_2)}\leq \max\left\{\frac{\ln \mathpzc{P}_1(\omega_1)}{\mathpzc{T}_1(\omega_1)}, \frac{\ln \mathpzc{P}_1(\omega_2)}{\mathpzc{T}_1(\omega_2)}\right\}.
\end{equation}
Thus, there exits a $\omega'\in \{\omega_1, \omega_2\}$ such that $F_1(\omega')\geq F_1(\omega)$. If $\mathpzc{T}_1(\omega_1)=0$, we have $\omega_1\in \Omega_0$ and therefore, by \eqref{equation: Omega-pstk1}, $\mathpzc{P}_1(\omega_1)\leq 1$. This implies $\mathpzc{P}_1(\omega)\leq \mathpzc{P}_1(\omega_2)$ and $\mathpzc{T}_1(\omega)=\mathpzc{T}_1(\omega_2)$, hence $F_1(\omega_2)\geq F_1(\omega)$. Finally, if  $\mathpzc{T}_1(\omega_1)<0$, take the sequence $\bar{\omega}=(\omega_1, \ldots, \omega_1, \omega_2, \ldots, \omega_2)$ in which the number of $\omega_1$ is $\mathpzc{T}_1(\omega_2)$ and the number of $\omega_2$ is $-\mathpzc{T}_1(\omega_1)$. Direct calculation shows $\mathpzc{T}_1(\bar{\omega})=0$, so that $\bar{\omega}\in \Omega_0$. Consequently, by \eqref{equation: Omega-pstk1},
$(\mathpzc{P}_1(\omega_1))^{\mathpzc{T}_1(\omega_2)}\cdot (\mathpzc{P}_1(\omega_2))^{-\mathpzc{T}_1(\omega_1)}\leq 1.$ Taking the logarithm of both sides, we get 
\begin{equation}\label{equation: P1T1P1T2}
\mathpzc{T}_1(\omega_2)\ln \mathpzc{P}_1(\omega_1)-\mathpzc{T}_1(\omega_1)\ln \mathpzc{P}_1(\omega_2)\leq 0.
\end{equation}
This, together with \eqref{equation: P1omega12}, implies $\frac{\ln \mathpzc{P}_1(\omega_2)}{\mathpzc{T}_1(\omega_2)}\geq \frac{\ln \mathpzc{P}_1(\omega)}{\mathpzc{T}_1(\omega)}$ or, equivalently, $F_1(\omega_2)\geq F_1(\omega)$.

To summarize, for any $\omega\in  \Omega_{12}\notin \Omega_{1p}$, we have obtained  a proper subsequence $\omega'\in \Omega_{12}$ of $\omega$ such that $F_1(\omega')\geq F_1(\omega).$ Since $\omega$  is a finite set, iterating the above procedure will eventually give rise to a subsequence $\hat{\omega}\in \Omega_{1p}$ of $\omega$ such that $F_1(\hat{\omega})\geq F_1(\omega).$
\end{proof}

Furthermore, we have the following result
\begin{lemma}\label{lemma: pi1 T1P1 leq 1}
$\pi_1^{-\mathpzc{T}_1(\omega)}\mathpzc{P}_1(\omega)\leq 1 \text{ for all }\omega\in \Omega_1.$
\end{lemma}
\begin{proof}
The proof is the same as that of Lemma~\ref{equation: sk beta ds} and is omitted here.
\end{proof}

By induction, suppose we have determined $\pi_{r-1}$, $r\geq 2$. To define $\pi_r$, let $\Omega_r$ be the set of $\mathds{S}_r$-balanced sequences and
\begin{equation}\label{equation: pstk2}
p_{st}^{k, r}=
\begin{cases}
p_{st}^{k, r-1}/\pi_{r-1}, & \mbox{if } s=r-1 \\
p_{st}^{k, r-1}, & \mbox{otherwise}.
\end{cases}
\end{equation}
For a sequence $\omega=\left(x_{s_it_i}^{k_i}, x_{s'_it'_i}^{k'_i}\right)_{i=1}^n$, let 
\begin{equation}\label{equation: P2}
\begin{aligned}
\mathpzc{P}_r(\omega)=\prod_{i=1}^n \frac{p_{s_it_i}^{k_i, r}}{p_{s'_it'_i}^{k'_i, r}}.
\end{aligned}
\end{equation}
If $\omega\in \Omega_{r-1}$, we have, by Lemma~\ref{lemma: pi1 T1P1 leq 1}, $\mathpzc{P}_r(\omega)\leq 1$. Let $\Omega_{r1}=\{\omega\in \Omega_r: \mathpzc{T}_r(\omega)<0\}$ and $\Omega_{r2}=\{\omega\in \Omega_r: \mathpzc{T}_r(\omega)>0\}$. Define
\begin{equation}\label{equation: pi2 DEU}
\pi_r=
  \begin{cases}
    1, & \mbox{if } \Omega_{r1}= \varnothing\text{ and } \Omega_{r2}= \varnothing \\
   \inf_{\omega\in \Omega_{r1}}[\mathpzc{P}_r(\omega)]^{1/\mathpzc{T}_r(\omega)}, & \mbox{if }\Omega_{r1}\neq \varnothing\text{ and } \Omega_{r2}= \varnothing \\
    \sup_{\omega\in \Omega_{r2}}[\mathpzc{P}_r(\omega)]^{1/\mathpzc{T}_r(\omega)}, & \mbox{if } \Omega_{r2}\neq \varnothing. 
  \end{cases}
\end{equation}
Similar to the case of $\pi_1$, we can show $\pi_r\in (0, \infty)$ and 
\begin{equation}\label{equation: Omega1pi2}
 \pi_r^{-\mathpzc{T}_r(\omega)}\mathpzc{P}_r(\omega)\leq 1 \text{ for all }\omega\in \Omega_r.
\end{equation}

With $\beta$ and $(\pi_1, \ldots, \pi_S)$ thus defined, to show $\mathcal{O}$ is DEU rational, it suffices to show it is OEU rational. Take any $\omega=\left(x_{s_it_i}^{k_i}, x_{s'_it'_i}^{k'_i}\right)_{i=1}^n\in \Omega$. Since $\Omega_S=\Omega$,  there exists an $r\in \mathds{S}\cup\{0\}$ such that $\omega\in \Omega_r\backslash \Omega_{r-1}$, where we set $\Omega_{-1}=\varnothing$. Then we have, by \eqref{equation: Omega1pi2},
$$\prod_{i=1}^n \frac{p_{s_it_i}^{k_i}/(\pi_{s_i}\beta^{t_i})}{p_{s'_it'_i}^{k'_i}/(\pi_{s'_i}\beta^{t'_i})}= \pi_r^{-\mathpzc{T}_r(\omega)}\mathpzc{P}_r(\omega)\leq 1.$$
It follows from Lemma~\ref{lemma: OEU rational} that $\mathcal{O}$ is OEU rational, hence DEU rational. 

\section{Implementation of the Revealed Preference Test for AS's experimental data}\label{section: Implementation of the Revealed Preference Test }
The CTB decision problem in AS's experiment is much easier than the one in Section~\ref{section:The Setup} and a test simpler than the one in Section~\ref{section: The Case with Objective Probabilities} is available. Specifically, AS's CTB decision is to seek $(x_t, x_{t+d})$ which solves
\begin{equation}\label{equation: OPAS}
\begin{cases}
  \max & \beta^t \pi_1 u(x_{t}) +\beta^{t+d}\pi_2u(x_{t+d})\\
  \mbox{s. t.} &p_t x_t+p_{t+d}x_{t+d}=I.
\end{cases}
\end{equation}
Suppose we have a set of $K$ observations $\mathcal{O}=\left\{(p_t^k, p_{t+d}^k; x_t^k, x_{t+d}^k;\pi_1^k, \pi_2^k): k=1, 2, \ldots, K\right\}$. This set is DEU rational if there exist a discount factor $\beta$ and a utility function $u$ such that $(x_t^k, x_{t+d}^k)$ solves Problem~\eqref{equation: OPAS} when $(p_t, p_{t+d})=(p_t^k, p_{t+d}^k)$ and $(\pi_1, \pi_2)=(\pi_1^k, \pi_2^k)$. 

To determine the DEU rationality of $\mathcal{O}$, we follow the procedure in Section~\ref{section: The Case with Objective Probabilities}. For the sake of completeness, let us spell out the details. A sequence $\omega=\left\{(x_{t_i}^{k_i}, x_{t'_i}^{k'_i}\right\}_{i=1}^n$ is balanced if  each $k$ appears as $k_i$ the same number of times it appears as $k'_i$. Let $\rho_{t}^k=p_{t}^k/\pi_1^k$, $\rho_{t+d}^k=p_{t+d}^k/\pi_2^k$, and
\begin{equation*}
\mathpzc{T}_0(\omega)=\sum_{i=1}^{n}(t_i-t'_i)\text{ and }\mathpzc{P}_o(\omega)=\prod_{i=1}^n \frac{\rho_{t_i}^{k_i}}{\rho_{t'_i}^{k'_i}}.
\end{equation*}
The data set $\mathcal{O}$  satisfies SAR-ODEU if $\mathpzc{P}_o(\omega)\leq 1$ for any balanced sequence $\omega$ with $\mathpzc{T}_0(\omega) \geq 0$. Following the argument in Section~\ref{section: Proof of Theorem SAR-ODEU}, we can deduce that $\mathcal{O}$ is DEU rational if and only if it satisfies SAR-ODEU.

How to verify the condition SAR-ODEU? To answer this question, let $\mathds{T}=\{7, 7+28, 7+56\}$ and 
$$\sum=\left\{(k, t, k', t')\in \mathds{K}\times \mathds{T}\times \mathds{K}\times\mathds{T}: x_{t}^k>x_{t'}^{k'}\right\}.$$
Let $N$ be the cardinality of $\sum$. Construct a matrix $G_1$ of dimension $K\times N$, each element $\sigma$ in $\sum$ corresponding to a column. The entry in row $\hat{k}$ and column $\sigma=(k, t, k', t')$ is $1$ if $\hat{k}=k$, $-1$ if $\hat{k}=k'$, and $0$ otherwise. Construct a row vector $G_2$ of dimension $N$ such that its entry corresponding to $\sigma=(k, t, k', t')\in \sum$ is $t-t'$. Finally, let $\delta$ be another row vector of dimension $N$ whose entry corresponding to $\sigma=(k, t, k', t')\in \sum$ is $\ln \rho_{t}^{k}-\ln  \rho_{t'}^{k'}$. Then SAR-ODEU holds for $\mathcal{O}$ if and only if the optimal value of the following linear program is less than or equal to zero \citep[][Appendix C]{Echenique2015}:
$$\begin{cases}
  \max_v &\delta\cdot v\\
  \mbox{s. t.} &G_1 v=0\\
  & G_2\cdot v\geq 0,\\
  &\mathbf{e}\cdot v=1,\\
  &v\geq 0,
\end{cases}$$
where $\mathbf{e}\in \mathbb{R}^N$ is a vector of all ones and the third constraint serves as a normalization. We use the Matlab function {\sffamily linprog} to solve this problem.

Note that each CTB decision in AS's experiment takes the form of a tuple $\gamma=(d, a_t, c_t, \pi_1, \pi_2)$. To carry out the above test, we need to translate this tuple into a price-consumption pair. For the consumptions, let $x_{t}=a_t\cdot c_t$ and $x_{t+d}=a_{t+d}(100- c_t)$; for the prices, take the dollar in date $(t+d)$ as the numeraire, so that $p_{t+d}=1$ and $p_{t}=a_{t+d}/a_t$. In this way, we can transform the set $\Gamma$ of CTB decisions of each subject into the form of a data set $\mathcal{O}$, and then by applying the above test, we can determine its DEU rationality.

\end{appendices}

\bibliographystyle{plainnat}
\bibliography{library}
\end{document}